\documentclass[
  aps, 
  prl, 
  reprint,
  nobibnotes 
]{revtex4-2}

\usepackage{amssymb}
\usepackage{amsmath}
\usepackage{mathtools}

\usepackage{amsthm}
\newtheorem{theorem}{Theorem}

\usepackage{adjustbox}

\usepackage[hidelinks]{hyperref}

\usepackage{tikz}
\usetikzlibrary{
  cd,
  arrows.meta,
  backgrounds,
  snakes
}

\tikzcdset{
  arrow style=tikz, 
  diagrams={
    >={
      Computer Modern Rightarrow[
        length=4pt, width=4pt
      ]
    },
    shorten=0pt
  }
}

\definecolor{darkblue}{rgb}{0.05,0.25,0.65}
\definecolor{darkgreen}{RGB}{20,140,10}
\definecolor{lightgray}{rgb}{0.9,0.9,0.9}
\definecolor{darkorange}{RGB}{200,100,5}
\definecolor{darkyellow}{rgb}{.91,.91,0}
\definecolor{lightolive}{RGB}{225, 220, 185}

\newcommand{\acts}{\hspace{-.5pt}\raisebox{1.4pt}{\;\rotatebox[origin=c]{90}{$\curvearrowright$}}\hspace{.6pt}}

\newcommand{\GeneralBrillouinTorus}[1]{\widehat{T}^{#1}}

\newcommand{\BrillouinTorus}{\GeneralBrillouinTorus{2}}

\newcommand{\ValenceBundle}{\mathcal{V}}

\newcommand{\HilbertSpace}{\mathcal{H}}

\newcommand{\ChernNumber}{C}

\newcommand{\ActedIndexedOperator}[4]{          {#1}
          _{\hspace{-1.5pt}
           \scalebox{.6}{${#4}\hspace{-2pt}
           \left[
           \def\arraystretch{.8}
           \def\arraycolsep{0pt}
           \begin{array}{c}
            {#2} \\ #3
            \\[-12pt]
            {}
           \end{array}
           \right]
           $}
        }
}

\newcommand{\ActedWOperator}[3]{
  \ActedIndexedOperator{\widehat{W}}{#1}{#2}{#3}
}

\newcommand{\WOperator}[2]{
  \ActedWOperator{#1}{#2}{}
}

\newcommand{\ZetaOperator}{\widehat{\rule{0pt}{6pt}\smash{\zeta}}}

\begin{document}


\author{Hisham Sati}
\email{hsati@nyu.edu}
\affiliation{Center for Quantum and Topological Systems, Division of Science, New York University, Abu Dhabi, UAE}

\author{Urs Schreiber}
\email{us13@nyu.edu}
\affiliation{Center for Quantum and Topological Systems, Division of Science, New York University, Abu Dhabi, UAE}

\title{
  Identifying Anyonic Topological Order in
  \\
  Fractional Quantum Anomalous Hall Systems
}

\begin{abstract}
  Recently observed fractional quantum \emph{anomalous} Hall materials (FQAH) are candidates for topological quantum hardware, but their required anyon states are elusive.
  We point out dependence on monodromy in the \emph{fragile} band topology in 2-cohomotopy. An algebro-topological theorem of Larmore \& Thomas (1980) then identifies FQAH anyons over momentum space. Admissible braiding phases are $2\ChernNumber$-th roots of unity, for $C$ the Chern number.
  This lays the foundation for understanding symmetry-protected topological order in FQAH systems, reducing the problem to computations in equivariant cohomotopy.
\end{abstract}

\date{\today}

\maketitle


\section{Motivation and Introduction}

The holy grail of quantum materials research is arguably
(cf. \cite{Wang2020}) the understanding and manipulation of:
\begin{itemize}

\item[(i)] crystalline {\bf topological phases} of matter --- where electron Bloch states in gapped valence bands span topologically non-trivial vector bundles over Brillouin tori $\GeneralBrillouinTorus{d}$ of crystal momenta (cf. \cite{Sergeev2023}\cite[\S 2]{SS22-Ord}),

\item[(ii)] exhibiting {\bf topological order} --- where the system's ground state  is degenerate and gapped, whence effectively (far above the lattice scale) that of a topological field theory (cf. \cite{Stanescu2020}\cite{Simon2023}),

\item[(iii)] hosting {\bf anyons} --- whence homotopy classes of adiabatic deformations of the material, notably the \emph{braiding} of soliton worldlines,
act on these ground states as unitary Berry phases (cf. \cite[\S 3]{MySS2024}\cite[\S 3]{SS22-Ord}).
\end{itemize}

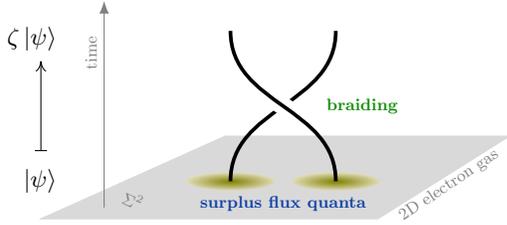
\begin{figure}[htbp]
  \centering
  \begin{tikzpicture}[
    xscale=.7,
    baseline=(current bounding box.center)
  ]
    \draw[
      gray!30,
      fill=gray!30
    ]
      (-4.6,-1.5) --
      (+1.8,-1.5) --
      (+1.8+3-.5,-.4) --
      (-4.6+3+.5,-.4) -- cycle;

    \begin{scope}[
      shift={(-1,-1)},
      scale=1.2
    ]
    \shadedraw[
      draw opacity=0,
      inner color=olive,
      outer color=lightolive
    ]
      (0,0) ellipse (.7 and .1);
    \end{scope}

    \draw[
     line width=1.4
    ]
      (-1,-1) .. controls
      (-1,0) and
      (+1,0) ..
      (+1,+1);

  \begin{scope}
    \clip 
      (-1.5,-.2) rectangle (+1.5,1);
    \draw[
     line width=7,
     white
    ]
      (+1,-1) .. controls
      (+1,0) and
      (-1,0) ..
      (-1,+1);
  \end{scope}
  
    \begin{scope}[
      shift={(+1,-1)},
      scale=1.2
    ]
    \shadedraw[
      draw opacity=0,
      inner color=olive,
      outer color=lightolive
    ]
      (0,0) ellipse (.7 and .1);
    \end{scope}
    \draw[
     line width=1.4
    ]
      (+1,-1) .. controls
      (+1,0) and
      (-1,0) ..
      (-1,+1);

  \node[
    rotate=-25,
    scale=.7,
    gray
  ]
    at (-2.85,-1.25) {
      $\Sigma^2$
    };

  \draw[
    -Latex,
    gray
  ]
    (-3.4,-1.35) -- 
    node[
      near end, 
      sloped,
      scale=.7,
      yshift=7pt
      ] {time}
    (-3.4, 1.4);

  \node[
    scale=.7
  ] at 
    (0,-1.3)
   {\bf \color{darkblue} surplus flux quanta};

  \node[
    scale=.7
  ] at 
    (1.5,0)
   {\bf \color{darkgreen} braiding};

\node [
 rotate=34,
 scale=.7
] at (3.15,-1.05) {
  \color{gray}
  2D electron gas
};

\node at (-4.3,+.9) {%
  \llap{$\zeta \, \vert \psi \rangle$}%
};

\draw[
  line width=.4,
  >={
    Computer Modern Rightarrow[
      length=4pt, width=4pt
    ]
  },
  |->,
]
  (-4.6, -.6) --
  (-4.6, +.6);

\node at (-4.3,-1) {%
  \llap{$\vert \psi \rangle$}%
};

\end{tikzpicture}

\caption{
  In an FQH system at unit \emph{filling fraction} $1/K$, each electron in the effectively 2-dimensional material is ``bound'' to $K$ quanta of transverse magnetic flux, as an effective result of strong interaction. On this backdrop, each surplus flux quantum is like the lack of $1/K$th of an electron, and as such called a \emph{quasi-hole}.   Just these quasi-hole surplus flux quanta are the (abelian) anyons in FQH systems (cf. \cite{Stormer99}), in that the quantum state of the system picks up a \emph{braiding phase} $\zeta \in \mathbb{C}^\times$ (a root of unity) when a pair of these swap positions.
}
\label{FigureA}
\end{figure}

Among candidates, \emph{fractional quantum Hall systems} (FQH, cf. \cite{Stormer99} \& Fig. \ref{FigureA}) stand out in that their (abelian) anyons are actually being observed in recent years, by independent groups and across different platforms (starting with \cite{Nakamura2020}, recent pointers in \cite{Veillon2024}). This makes FQH materials a candidate hardware for much anticipated topological quantum computers \cite{Stanescu2020}, plausibly necessary for future quantum computing of utility value.


While the large external magnetic fields required for typical FQH systems may hinder their practicality as hardware, it has rather recently been confirmed for various materials
\cite{Cai2023}\cite{Zeng2023}\cite{Park2023}\cite{Lu2024}
that there also exist ``anomalous'' forms of FQH materials (FQAH), where the role of the external magnetic field is instead played by magnetic properties intrinsic to crystalline topological phases of matter called a \emph{fractional Chern insulators} (FCI, \cite{Regnault2011}\cite{Neupert2015}, cf. Fig. \ref{FigureB}).

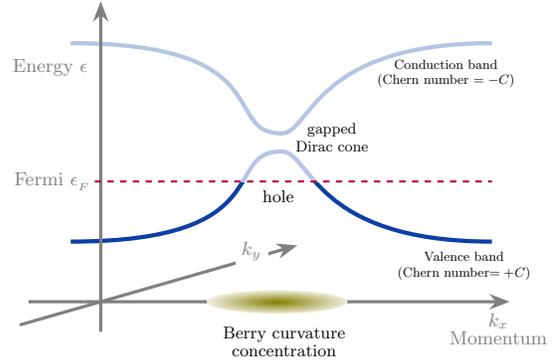
\begin{figure}[h]
\adjustbox{
  scale=0.8
}{
\begin{tikzpicture}[
    >=Stealth,
]

\draw[
  line width=2,
  darkblue
]
  (-.5,-2) .. controls
  (3,-2) and
  (2,-.5) ..
  (3, -.5) .. controls
  (3.5,-.5) and
  (3.5,-2) ..
  (6.5,-2);

\begin{scope}[
  shift={(0,-.7)},
  yscale=-1
]
\draw[
  line width=2,
  darkblue
]
  (-.5,-2) .. controls
  (3,-2) and
  (2,-.5) ..
  (3, -.5) .. controls
  (3.5,-.5) and
  (3.5,-2) ..
  (6.5,-2);
\end{scope}

\draw[
  draw opacity=0,
  fill=white,
  fill opacity=.7
]
  (-.5,-1) rectangle
  (6.5,1.5);
  
\draw[
  ->,
  gray,
  line width=1.5
]
  (0,-3.5) --
    node[
      pos=.8,
      xshift=-3pt
    ] {
      \llap{
        Energy 
        $\epsilon$
      }
    }
  (0,2);

\draw[
  ->,
  gray,
  line width=1.5
]
  (-1.2,-3) --
    node[
      pos=1.1,
      yshift=-14pt
    ] {
      \llap{
        \def\arraystretch{.7}
        \begin{tabular}{c}
          $k_x$
          \\
          Momentum 
        \end{tabular}
      }
    }
  (6.8,-3);

\draw[
  ->,
  gray,
  line width=1.5,
  shift={(0,-3)}
]
  (16:-1.4cm) --
    node[
      pos=.85,
      shift={(-2pt,+3pt)}
    ]{
      \colorbox{white}{$k_y$}
    }
  (16:+3.4cm);

\draw[
  line width=1,
  color=purple,
  dashed
]
  (-.1,-1) --
  (6.5,-1);
\node[
  gray
] at  (-0.05,-1) {
  \llap{
  Fermi $\epsilon_{{}_{F}}$
  }
};

\node[scale=.7] at 
  (6,-2.4) {
    \def\arraystretch{1}
    \begin{tabular}{c}
      Valence band
      \\
      (Chern number$=+C$)
    \end{tabular}
  };

\node[scale=.7] at 
  (5.7,+.8) {
    \def\arraystretch{1}
    \begin{tabular}{c}
      Conduction band
      \\
      (Chern number $= -C$)
    \end{tabular}
  };

\node[
  scale=.8
] at (3.8,-.3) {
  \def\arraystretch{1}
  \begin{tabular}{c}
    gapped 
    \\
    Dirac cone
  \end{tabular}
};

\node[
  scale=.9
] at 
  (2.95,-1.25) {hole};

\shadedraw[
  draw opacity=0,
  inner color=olive,
  outer color=lightolive!50,
  shift={(2.93,-3)},
  scale=.65
]
  (0,0) ellipse (1.8 and .3);

\node[
  scale=.9
] at (3,-3.7) {
  \def\arraystretch{.9}
  \begin{tabular}{c}
    Berry curvature 
    \\
    concentration
  \end{tabular}
};

\end{tikzpicture}
}
\caption{
  Fractional Chern insulators (FCI) exhibiting  
  a fractional quantum anomalous Hall effect (FQAH) are typically realized 
  \cite[\S II.A]{Chang2023}
  by gapping Dirac cones of 2D 2-band systems.
}
\label{FigureB}
\end{figure}

Concretely, the hallmark properties of anomalous Hall systems are all due to the existence of non-vanishing \emph{Berry curvature}
$\Omega(\vec k)$ over the crystal's Brillouin torus of Bloch momenta,  $ \vec k \in \BrillouinTorus$  (cf. \cite{Sinitsyn2007}\cite{Nagaosa2010}\cite[\S III.D]{Xiao2010}\cite[\S 5.1]{Vanderbilt2018}), a measure for the electron state's dependence on their crystal momenta which, takes the role of magnetic flux density $F(\vec x)$ in the space of positions $\vec x$ inside FQH systems. 

This analogy is quite strong; we suggest to think of it as a \emph{duality} (cf. Fig. \ref{PositionMomentumDuality}): The \emph{Lorentz force} law which underlies the ordinary Hall effect and the \emph{anomalous velocity law} which underlies the anomalous Hall effect are mirror images of each other under exchanging position space with momentum space and magnetic flux density with Berry curvature (an observation that goes back to \cite{Chang1995} \cite{Chang1996}\cite[\S III.B]{Sundaram1999}, 
review in \cite{Xiao2005}\cite[(69)]{Sinitsyn2007}\cite[\S III.A]{Xiao2010}\cite{Chong2010}\cite[(12-13)]{Stephanov2012}), see Fig. \ref{PositionMomentumDuality}. This gives a good understanding of the anomalous Hall effect as a Hall effect in momentum space.

\begin{figure}[h]
\adjustbox{
  scale=0.9,
  raise=-1.6cm,
  fbox
}{
\begin{tikzpicture}

\node[scale=.7] at (-2.9,1.2) {%
  \color{gray}%
  \bf%
  \def\arraystretch{.9}%
  \begin{tabular}{@{}c@{}}%
    Semiclassical 
    \\
    equation of motion
    \\
    of crystal electrons
  \end{tabular}
};

\node at (0,0) {
\begin{tikzcd}[
  column sep=-2pt,
  row sep=10pt
]
  \dot {\vec k}
  &=&
  -
  \partial \epsilon
  /
  \partial\vec x
  &+&
  \overset{
    \mathclap{
      \adjustbox{
        scale=.7,
        raise=2pt
      }{
        \color{darkblue}
        \bf
        Lorentz force
      }
    }
  }{
    \vec x \,\times\, e\vec B
  }
  &\phantom{---}&
  \overset{
    \mathclap{
      \adjustbox{
        scale=.7,
        raise=4pt
      }{
        \color{darkblue}
        \bf
        \def\arraystretch{.9}
        \begin{tabular}{c}
          position
        \end{tabular}
      }
    }
  }{
    \vec x
  }
  \ar[
    <->,
    shorten=-2pt,
    dashed,
    gray,
    d
  ]
  &\phantom{---}&
  \overset{%
    {%
      \adjustbox{
        scale=.7,
      }{%
        \color{darkblue}%
        \bf%
        \def\arraystretch{.9}%
        \begin{tabular}{@{}c@{}}%
          magnetic
          \\
          flux density
        \end{tabular}%
      }%
    }%
  }{
    e \vec B
  }
  \ar[
    <->,
    shorten=-2pt,
    dashed,
    gray,
    d
  ]
  \\
  \dot {\vec x}
  &=&
  \phantom{+}
  \partial \epsilon
  /
  \partial\vec k
  &-&
  \underset{
    \mathclap{
      \adjustbox{
        scale=.7,
        raise=-2pt
      }{
        \color{darkblue}
        \bf
        anomalous velocity
      }
    }  
  }{
    \vec k \,\times\, \phantom{e}\vec\Omega
  }
  &&
  \underset{
    \mathclap{
      \adjustbox{
        scale=.7,
        raise=-3pt
      }{
        \color{darkblue}
        \bf
        \def\arraystretch{.9}
        \begin{tabular}{c}
          momentum
        \end{tabular}
      }
    }
  }{
    \vec k
  }
  &&
  \underset{
    \mathclap{
      \adjustbox{
        scale=.7,
        raise=-6.5pt
      }{
        \color{darkblue}
        \bf
        \def\arraystretch{.85}
        \begin{tabular}{c}
          Berry 
          \\
          curvature
        \end{tabular}
      }
    }
  }{
    \vec \Omega
  }
\end{tikzcd}
};
\begin{scope}[
  shift={(-2.3,0)}
]
\draw[
  <->,
  densely dashed
]
  (2.25,.5) .. controls
  (3.2,.5) and
  (3.2,-.5) ..
  (2.25,-.5);
\end{scope}

\node[scale=.8] at (2.6,1.5) {
  \color{darkorange}
  \bf
  Hall effect
};
\node[scale=.8] at (2.8,-1.5) {
  \color{darkorange}
  \bf
  \llap{Anomalous }Hall effect
};

\end{tikzpicture}
}
\caption{
  There is a duality between the FQH effect and its anomalous FQAH version, under which the ordinary space of positions inside the 2D material is exchanged for the ``reciprocal'' space of crystal momenta, while the external magnetic flux density is exchanged for the Berry curvature of the Bloch bands over this momentum space.
}
\label{PositionMomentumDuality}
\end{figure}
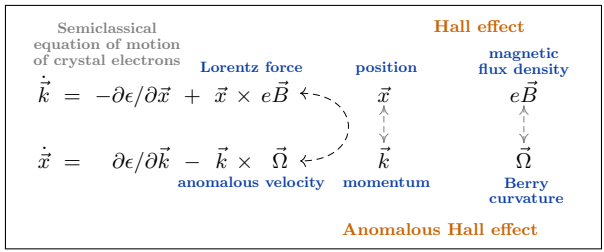
This duality is particularly remarkable for our purpose, because the full definition of flux in position space is 
\begin{itemize}
\item[(i.)] \cite{SS25-Flux} known  to require choices of \emph{flux-quantization laws} which exhibit the total flux -- hence dually: the Chern number \eqref{TheChernNumber} -- as a ``rationalized'' shadow of maps to a classifying space, cf. \eqref{TheChernNumber}, while
\item[(ii.)] \cite{SS25-FQH} exotic (below: \emph{fragile}) flux-quantization laws were recently understood to control the existence and nature of anyonic solitons in FQH systems.
\end{itemize}

Just that identification of anyonic solitons in the dual FQAH systems had remained wide open but needs addressing for these systems to serve as topological quantum hardware.

We now present a resolution by:
\begin{itemize}
\item[(i.)] pointing out that the \emph{fragile} band topology of (fractional) Chern insulators (while often conflated with the \emph{stable} band topology) influences the adiabatic Berry holonomy \eqref{AdiabaticHolonomy} acting on the ground states,

\item [(ii.)] computing the effect, verifying that it exhibits anyonic topological order over the Brillouin torus (Thm. \ref{DerivingTorusAlgebra} below).
\end{itemize}

We use basics of algebraic topology 
\cite{Hatcher2002}; for exposition in our context see \cite{Sergeev2023}\cite[\S A.3]{SS25-FQH}\cite[\S A]{SS25-Srni}\cite[\S A.2]{SS24-Obs}.

\section{Methods and Results}

First to note that generic fractional Chern insulators are (cf. \cite[\S II.A.1]{Chang2023}) 2D 2-band systems
whose Bloch Hamiltonian over the Brillouin torus $\BrillouinTorus$ of crystal momenta may be expanded in Pauli matrices $\{\sigma_i\}_{i = 1}^3$ as:
\begin{equation}
\label{BlochHamiltonian}
\begin{aligned}
    H_{\mathrm{Blch}}
    & :  
  \begin{tikzcd}[
    ampersand replacement=\&
  ]
    \BrillouinTorus
    \ar[rr]
    \&\&
    \mathrm{Mat}_{2 \times 2}(\mathbb{C})
  \end{tikzcd}
  \\
  H_{\mathrm{Blch}}(\vec k)
  & =
  h_0(\vec k) 
    + 
  \underbrace{
  \textstyle{\sum_{i = 1}^3}
  h_i(\vec k)
  \,
  \sigma_i
  }_{
    H(\vec k)
  }
  \,,
  \;\;\;
  h_{{}_{(-)}}(-) \in \mathbb{R}
  \,,
\end{aligned}
\end{equation}
where the system is \emph{gapped} in that the difference of the two eigenvalues is non-vanishing:
\begin{equation}
  \big\vert H(\vec k) \big\vert
  \;:=\;
  \sqrt{
    \textstyle{\sum_{i=1}^3}
    h_i(\vec k)
  }
  \;\;
  >
  \;\;
  0
  \,,
\end{equation}
whence the \emph{valence bundle} $\mathcal{V}$ of lower eigenspaces (cf. Fig. \ref{FigureB}) is the kernel bundle over $\BrillouinTorus$ of the projectivized operators:
\begin{equation}
  \label{TheValenceBundle}
  \ValenceBundle
  \;\simeq\;
  \mathrm{ker}\big(
    \mathrm{id} + H/\vert H \vert
  \big)
  \,,
\end{equation}
which shows that the valence bundle is exactly encoded in a map $\vec h/\vert H \vert$ to the 2-sphere (cf. \cite[(8.3-4)]{Sergeev2023}):
\begin{equation}
  \label{TheMapToTheSphere}
  \begin{tikzcd}
    \BrillouinTorus
    \ar[
      rr,
      "{
        \vec h/\vert H \vert
      }"
    ]
    \ar[
      rrrr,
      rounded corners,
      to path={
        (\tikztostart.north)
        -- ++(0, +6pt)
        -| 
        node[
          pos=.27,
        ] {
          \adjustbox{
            bgcolor=white,
            scale=.7
          }{
            $H/\vert H \vert$
          }
        }
        ([yshift=+00pt]\tikztotarget.north)
      }
    ]
    &&
    S^2
    \ar[
      rr,
      "{
        \sigma
      }",
      "{
        \vec x 
        \;\mapsto\;
        \sum_i x^i \sigma_i
      }"{swap}
    ]
    &&
    \mathrm{Mat}_{2\times 2}(\mathbb{C})
    \,.
  \end{tikzcd}
\end{equation}

More concretely, the lower eigenspace bundle of (what we may identify as the ``universal normalized 2-band Bloch Hamiltonian'', denoted) $\sigma$ in \eqref{TheMapToTheSphere} is the \emph{tautological} complex line bundle on the 2-sphere, the one with unit Chern class class (this is \cite[Lem. 4.5.12]{SS25-Bund}):
\begin{equation}
  \mathrm{ker}\big(
    \mathrm{id} + \sigma
  \big)
  \;=\;
  1 
  \;\in\;
  \mathbb{Z}
  \;\simeq\;
  H^2\big(
    S^2;\,
    \mathbb{Z}
  \big)
  \,,
\end{equation}
whence it is the pullback of the universal complex line bundle on infinite projective space along the canonical inclusions (cf. \cite[\S 3.8]{SS23-MF}):
\begin{equation}
  \label{SequenceOfProjectiveSpaces}
  \begin{tikzcd}[
    column sep=8pt,
    shorten=-2pt
  ]
    \iota
    :
    S^2 
    \simeq
    \mathbb{C}P^1
    \ar[r, hook]
    &
    \mathbb{C}P^2
    \ar[rr, hook, densely dotted]
    &&
    \mathbb{C}P^\infty
    \,:=\,
    \displaystyle{\bigcup_{\mathclap{N \to \infty}}}
    \;
    \mathbb{C}P^N
    \!,
  \end{tikzcd}
\end{equation}
 of complex projective \emph{Grassmannian} moduli spaces of complex lines (fibers of the valence bundle) inside a fixed $\mathbb{C}^{1+N}$ (the fiber of a trivial ambient Bloch bundle).

The infinite projective space $\mathbb{C}P^\infty$
\eqref{SequenceOfProjectiveSpaces} is the \emph{classifying space} for (complex line bundles and) integral 2-cohomology, meaning that the homotopy classes of maps into it, hence the connected components $\pi_0$ of the mapping space $\mathrm{Map}(-,-)$,  are (cf. \cite[Ex. 2.1]{FSS23-Char}):
\begin{equation}
  \pi_0 \, \mathrm{Map}\big(
    -
    ,\,
    \mathbb{C}P^\infty
  \big)
  \;\;
  \simeq
  \;\;
  H^2\big(
    -
    ;\,
    \mathbb{Z}
  \big)
  \,,
\end{equation}
and the corresponding class of the valence bundle \eqref{TheValenceBundle}, via \eqref{TheMapToTheSphere} and \eqref{SequenceOfProjectiveSpaces}, is its (first and only) \emph{Chern number} $C$:
\begin{equation}
  \label{TheChernNumber}
  \begin{tikzcd}[
    column sep=6pt,
    shorten=-2pt
  ]
    \Big[
    \BrillouinTorus
    \ar[
      rr,
      "{
        \frac
          {\vec h}
          {\rule{0pt}{4pt}\vert H \vert}
      }"{yshift=1pt}
    ]
    &&
    S^2
    \ar[r, hook]
    &
    \mathbb{C}P^\infty
    \Big]
  \,=\,
  C
  \,\in\,
  \mathbb{Z}
  \,\simeq\,
  H^2\big(
    \BrillouinTorus
    ;\,
    \mathbb{Z}
  \big)
  \,,
  \end{tikzcd}
\end{equation}
characterizing the \emph{topological phase} of the Chern insulator.

Now we come to the crucial subtlety not previously fully appreciated:

1. The valence bundle arises in \eqref{TheValenceBundle} as a sub-bundle of the trivial $\mathrm{rank}=2$ bundle and \emph{as such} is classified by the Grassmannian $\mathbb{C}P^1 \simeq S^2$, reflecting equivalence of topological phases under paths of deformations of the system small enough for the valence bundle not to mix with higher conduction bands. We note here that also $\mathbb{C}P^1 \simeq S^1$ classifies a cohomology theory, namely an  exotic ``nonabelian'' cohomology theory called \emph{2-cohomotopy} $\pi^2$ \cite[\S VII]{STHu59}\cite[Ex. 2.7]{FSS23-Char}:
\begin{equation}
  \pi_0
  \,
  \mathrm{Map}\big(
    -
    ;\,
    S^2
  \big)
  \;\;
  =:
  \;\;
  \pi^2(-)
  \,.
\end{equation}
Therefore, 
in the language of \cite{Bouhon2020} (following \cite{PoHarukiVishwanath2018}, cf. \cite{Bouhon2023})
2-cohomotopy classifies the \emph{fragile} band topology of generic FCIs: fragile, since its classification may in principle be broken by more drastic deformations.

2. But the class in \eqref{TheChernNumber} is that classified by $\mathbb{C}P^\infty$ 
\eqref{SequenceOfProjectiveSpaces}
and thus reflects equivalence under paths of exactly such more drastic deformations that may involve an arbitrary number of higher conduction bands, preserving only the \emph{stable} band topology.

3. This situation is subtle and prone to be underappreciated because, in the present case of 2D materials, it \emph{makes no difference for the topological charge}: The classical \emph{Hopf degree theorem} (cf. \cite[(35)]{FSS20-H}) says that the $n$-cohomotopy of closed $n$-manifolds coincides with their integral $n$-cohomology:
\begin{equation}
  \begin{tikzcd}[
    column sep=10pt, 
    shorten=-2pt,
    row sep=5pt
  ]
    \pi^2\big(
      \BrillouinTorus
    \big)
    \ar[
      rr,
      "{ \sim }"
    ]
    \ar[d, equals]
    &&
    H^2\big(
      \BrillouinTorus
      ;\,
      \mathbb{Z}
    \big)
    \ar[d, equals]
    \\
    \pi_0
    \mathrm{Map}\big(
      \BrillouinTorus
      ,\,
      S^2
    \big)
    \ar[
      rr,
      "{ \iota_\ast }"
    ]
    &&
    \mathrm{Map}\big(
      \BrillouinTorus
      ,\,
      \mathbb{C}P^\infty
    \big)
    \,,
  \end{tikzcd}
\end{equation}
whence authors may and do routinely address the Chern number \eqref{TheChernNumber} as the \emph{winding number} of the map $\vec h/ \vert H \vert$ \eqref{TheMapToTheSphere}, thereby tacitly conflating the fragile band topology in (what we identified as) $\pi^2$ with the stable band topology in $H^2(-;\mathbb{Z})$ --- which is harmless in the study of ordinary Chern insulators.

4. \emph{But} as we now turn attention to the fractional situation of Chern insulators exhibiting an FQAH effect and thus expected to exhibit anyonic topological order (\cite{Wen1991}\cite{SS22-Ord}), 
we highlight that the whole point is the possibility that --- while the external parameters of a system may adiabatically move through a closed path $\gamma$ of deformations, here: returning the deformation class $[\ValenceBundle]$ of the valence bundle back to itself --- the Hilbert space $\HilbertSpace_C$ of anyonic quantum ground states may be transformed by a non-trivial unitary $U_\gamma$ sensitive to (the homotopy class $[\gamma]$ of) the actual path (cf. \cite[\S II.A.2]{Nayak2008}\cite[\S 3]{MySS2024}).

In other words, the Hilbert space $\HilbertSpace_C$ of anyonic ground states observable under deformations that allow excursions (of the valence bundle $\ValenceBundle$ of Chern number $C$) into the first $N$ Bloch bands above the gap must be a representation $\rho$ of the fundamental group $\mathcolor{purple}{\pi_1}$ of maps to $\mathbb{C}P^N$:
\begin{equation}
  \label{AdiabaticHolonomy}
  \begin{tikzcd}[
    decoration=snake,
    sep=0pt
  ]
    \mathcolor{purple}{\pi_1}
    \Big(
    \mathrm{Map}\big(
      \BrillouinTorus
      ,\,
      \mathbb{C}P^N
    \big)
    ,\,
    C
    \Big)
    \ar[
      rr,
      "{
        \rho
      }"
    ]
    &&
    \mathrm{U}(\HilbertSpace_C)
    \\
    \hspace{3pt}%
      \mathclap{[\ValenceBundle]}%
    \hspace{3pt}%
    \ar[
      in=-39.5,
      out=180+39.5,
      looseness=10.4,
      decorate,
      -Latex,
      "{ \;\gamma }"{yshift=.5pt}
    ]
    &\longmapsto&
    \HilbertSpace
    \ar[
      in=-38,
      out=180+38,
      looseness=10,
      "{
        U_\gamma
      }"
    ]
  \end{tikzcd}
\end{equation}

Thereby, topological order is sensitive not just to the charge sector in $\pi_0$ of the mapping space, but also to $\pi_1$ of its connected components, and for that the difference between fragile and stable band topology makes all the difference --- this is our main result to be communicated here:

\noindent
\begin{theorem}\label{DerivingTorusAlgebra}
For fragile valence band topology, $N = 1$
\eqref{SequenceOfProjectiveSpaces}, the Hilbert spaces in \eqref{AdiabaticHolonomy} are modules over the non-commutative algebra
\begin{align}
  \label{TheHeisenbergGroupAlgebra}
  &
  \mathrm{Obs}_C
  \;\;
  :=
  \;\;
  \mathbb{C}\bigg[
    \pi_1
    \Big(
    \mathrm{Map}\big(
      \BrillouinTorus
      ,\,
      \mathcolor{purple}{S^2}
    \big)
    ,\,
    C
    \Big)
  \bigg]
  \\
  \notag
  &
  \simeq
  \Big\langle
    \WOperator{\mathcolor{purple}{1}}{0}
    ,
    \WOperator{0}{\mathcolor{purple}{1}}
    ,\,
    \underset{
      \mathclap{
        \adjustbox{
          scale=.6,
          rotate=40
        }
        {
          \llap{%
            \rm%
            central
            \hspace{-8pt}
          }
        }
      }
    }{
      \mathcolor{purple}{\ZetaOperator}
    }
    \;\,\Big\vert\;\,
    \WOperator{\mathcolor{purple}{1}}{0}
    \,
    \WOperator{0}{\mathcolor{purple}{1}}    
    \,=\,
    \mathcolor{purple}{\ZetaOperator}^{2}
    \,
    \WOperator{0}{\mathcolor{purple}{1}}    
    \,
    \WOperator{\mathcolor{purple}{1}}{0}
    ,\;
    \mathcolor{purple}{\ZetaOperator}^{2C} = 1
  \Big\rangle
  \,,
  \\[-11pt]
  \notag
\end{align}
where in each superselection sector (meaning: on irreducible representations) with 
$\ChernNumber \neq 0$ the operator $\ZetaOperator$
acts by multiplication with a $2\ChernNumber$-th root of unity,
\begin{equation}
  \label{TheRootOfUnity}
  \ZetaOperator
  \,
  \vert
  \psi \rangle
  \;=\;
  e^{ 
    \pi \mathrm{i} 
    \tfrac{p}{\ChernNumber} 
  }
  \,
  \vert \psi \rangle
  \,,
  \hspace{.5cm}
  p \in \big\{
    0, 1, \cdots, 
    \vert \ChernNumber\vert -1
  \big\}
  \,.
\end{equation}
On the other hand, for stable band topology, $N = \infty$, the central generator $\ZetaOperator$ disappears and the algebra becomes commutative with its two generators identified with a basis of torus 1-cycles:
\begin{align}
  \label{TheCommutativeAlgebra}
  &
  \mathbb{C}\bigg[
    \pi_1
    \Big(
    \mathrm{Map}\big(
      \BrillouinTorus
      ,\,
      \mathcolor{purple}{\mathbb{C}P^\infty}
    \big)
    ,\,
    C
    \Big)
  \bigg]
  \;\simeq\;
  \mathbb{C}\Big[
    H^1\big(
      \BrillouinTorus
      ;\,
      \mathbb{Z}
    \big)
  \Big]
  \\
  \notag
  &
  \simeq
  \Big\langle
    \WOperator{\mathcolor{purple}{1}}{0}
    ,
    \WOperator{0}{\mathcolor{purple}{1}}
    \;\,\Big\vert\;\,
    \WOperator{\mathcolor{purple}{1}}{0}
    \,
    \WOperator{0}{\mathcolor{purple}{1}}    
    \,=\,
    \WOperator{0}{\mathcolor{purple}{1}}    
    \,
    \WOperator{\mathcolor{purple}{1}}{0}
  \Big\rangle
  \,.
  \\[-11pt]
  \notag
\end{align}
\end{theorem}
\begin{proof}
  With the identification of \eqref{AdiabaticHolonomy} it is a question in pure algebraic topology to compute the fundamental groups of these mapping spaces. For \eqref{TheCommutativeAlgebra} this is elementary (cf. \cite[(15)]{SS25-FQH}), using that $\mathbb{C}P^\infty \simeq K(\mathbb{Z},2)$ is an Eilenberg-MacLane space. The computation for \eqref{TheHeisenbergGroupAlgebra} is much more intricate: It turns out to be a theorem of Larmore \& Thomas (1980) \cite[Thm. 1]{LarmoreThomas1980}, following Hansen (1974) \cite[Thm. 1]{Hansen1974} who obtained the result except for fixing the factor ``$2$'' in the exponent of $\ZetaOperator^{2C} = 1$. A more streamlined proof was later given by Kallel (2001) \cite[Prop. 1.5]{Kallel2001}. The observation that these purely algebro-topological results pertain to anyonic topological order of quantum Hall systems we recently highlighted in \cite[\S 3.3-4]{SS25-FQH}.

  With this, the statement \eqref{TheRootOfUnity} is a direct consequence of Schur's lemma, cf. \cite[(201)]{SS25-FQH}. (In the remaining case $\ChernNumber = 0$, a variant of \eqref{TheRootOfUnity} may be proven by invoking the modular equivariance expected of an anyon topological field theory, see \cite[\S 3.4]{SS25-FQH}.)
\end{proof}

But now we observe that the algebra \eqref{TheHeisenbergGroupAlgebra} is exactly the characteristic algebra of observables for anyons  on a torus (here: on the Brillouin torus of crystal momenta) with braiding phase $\zeta$:
This is due to \cite[(4.9)]{WenNiu1990}\cite[(4.14)]{IengoLechner1992}, reviewed in \cite[(4.21)]{Fradkin2013}\cite[(5.28)]{Tong2016}, cf. Fig. \ref{FigureTorusAlgebra}.

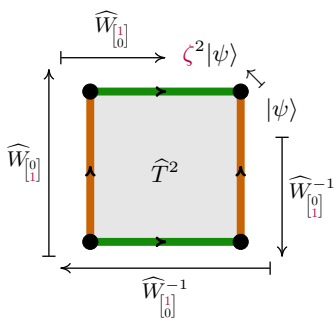
\begin{figure}[h]
\begin{tikzpicture}[
  >={
    Computer Modern Rightarrow[
      length=4pt, width=4pt
    ]
  }
]

\draw[
  fill=lightgray,
]
  (0,0) rectangle (2,2);

\node at (1,.8) {%
  \clap{\smash{$\BrillouinTorus$}}%
};

\draw[
  line width=3,
  color=darkorange
] 
  (0,0) --
  (0,2);
\draw[
  ->,
  line width=1
]
  (0,1) --
  (0,1.01);

\draw[
  line width=3,
  color=darkorange
] 
  (2,0) --
  (2,2);
\draw[
  ->,
  line width=1
]
  (2,1) --
  (2,1.01);

\draw[
  line width=3,
  color=darkgreen
] 
  (0,2) --
  (2,2);
\draw[
  ->,
  line width=1
]
  (1,2) --
  (1.01,2);

\draw[
  line width=3,
  color=darkgreen
] 
  (0,0) --
  (2,0);
\draw[
  ->,
  line width=1
]
  (1,0) --
  (1.01,0);

\draw[
  fill=black
] 
(0,0) circle (.1);
\draw[
  fill=black
] 
(2,0) circle (.1);
\draw[
  fill=black
] 
(0,2) circle (.1);
(2,0) circle (.1);
\draw[
  fill=black
] 
(2,2) circle (.1);

\node at (2-.4,2+.5) {
  $\mathcolor{purple}{\zeta}^2\vert \psi \rangle$
};

\node[
  rotate=+135  
] at (2+.2, 2+.2) {%
  \clap{$\mapsto$}%
};

\node at (2.55,2-.25) {
  $\vert \psi \rangle$
};

\draw[
  line width=.5,
  |->,
]
  (2.55, 1.4) --
  node[
    xshift=11.5pt,
    scale=.9
  ] {$
    \WOperator{0}{\mathcolor{purple}{1}}
      ^{-1}
  $}
  (2.55,-.2);

\draw[
  line width=.5,
  |->,
]
  (2.4, -.35) --
  node[
    yshift=-11pt,
    scale=.9
  ]{$
    \WOperator{\mathcolor{purple}{1}}{0}^{-1}
  $}
  (-.4, -.35);

\draw[
  line width=.5,
  |->,
]
  (-.55, -.2) --
  node[
    xshift=-9pt,
    scale=.9
  ] {$
    \WOperator{0}{\mathcolor{purple}{1}}
  $}
  (-.55,2.3);

\draw[
  line width=.5,
  |->,
]
  (-.4, 2.45) --
  node[
    scale=.9,
    yshift=12pt
  ]{$
    \WOperator{\mathcolor{purple}{1}}{0}
  $}
  (1, 2.45);
  
\end{tikzpicture}
\caption{
The relation in the algebra \eqref{TheHeisenbergGroupAlgebra} says that an adiabatic parameter loop around a basis of torus cycles induces transformation of the topological states by multiplication with the square of a complex phase
$\zeta$ to be identified \cite[Prop. 3.21]{SS25-FQH} with an anyon braiding phase (cf. Fig. \ref{FigureA}).
}
\label{FigureTorusAlgebra}
\end{figure}

In particular, from this observable algebra \eqref{TheHeisenbergGroupAlgebra} a variant of the Stone-von Neumann theorem theorem implies \cite[\S 3.4]{SS25-FQH} that the superselection sectors of $\HilbertSpace_C$ have dimension $\propto \mathrm{ord}(\zeta)$, thus exhibiting the ground state degeneracy characteristic of topological order.

\section{Conclusion and Outlook}

Highlighting that topological order is sensitive not just to the connected components $\pi_0$ of the moduli space of band topologies, but crucially to its fundamental groups $\pi_1$ (where the commonly conflated \emph{fragile} and \emph{stable} band topologies of Chern insulators may differ substantially) we have invoked mathematical results from algebraic topology to derive (predict) how, when the fragile topology of FQAH systems can be resolved in \emph{cohomotopy} $\pi^2$, then topological order becomes detectable with braiding phases of anyons localized over the reciprocal space of crystal momenta.
Curiously, this makes FQAH anyons be momentum-space duals of similarly cohomotopically classified anyons in FQH systems \cite{SS25-FQH}.

But realistic (fractional) Chern insulators of interested are typically subject to \emph{crystalline symmetry}, forcing the Bloch Hamiltonians \eqref{BlochHamiltonian} to transform under certain crystal point group symmetries $g \in G \acts \BrillouinTorus$ as (cf. \cite[\S 4.1]{Stanescu2020})
\begin{equation}
  \label{SymmetryOnBlochHamiltonian}
  H_{\mathrm{Blch}}(g \cdot \vec k) \,=\,  
  U_g(\vec k) 
   \circ
    H_{\mathrm{Blch}}(\vec k) 
   \circ
  U_{g^{-1}}(\vec k)
  \,,
\end{equation}
for compatible unitary operators
$$
  \begin{tikzcd}[sep=0pt]
    \BrillouinTorus \times G
    \ar[
      rr,
      "{  }"
    ]
    &&
    \mathrm{SU}(2)
    \\
    (\vec k,g) &\longmapsto&
    U_g(\vec k)
    \,.
  \end{tikzcd}
$$
In view of \eqref{TheMapToTheSphere} this means equivalently that the map characterizing the valence bundle is actually \emph{$G$-equivariant}
\begin{equation}
  \label{EquivarianceOfClassifyingMap}
  \begin{tikzcd}
    \BrillouinTorus
    \ar[
      in=60,
      out=180-60,
      looseness=5,
      shift right=4pt,
      "{ G \subset \mathrm{Pin}(2) }"{
        description
      }
    ]
    \ar[
      rr,
      "{
        \vec h / {\vert H \vert}
      }"
    ]
    &&
    \mathcolor{white}{\widehat{\mathcolor{black}{S^2}}}
    \smash{\mathrlap{\,.}}
    \ar[
      in=60,
      out=180-60,
      looseness=4.5,
      shift right=5pt,
      "{ G \subset \mathrm{Spin}(3) }"{
        description
      }
    ]
  \end{tikzcd}
\end{equation}
We may observe that, 
if and when the relevant deformations are all constrained to preserve this symmetry --- whence one speaks of the topological phases being \emph{symmetry-protected} (cf. \cite{Chiu2016}\cite[\S 2.3]{SS22-Ord})  --- the fragile topological moduli space of the system shrinks to the subspace $\mathrm{Map}(-,-)^{G} \subset \mathrm{Map}(-,-)$ of $G$-equivariant maps \eqref{EquivarianceOfClassifyingMap}, whose connected components $C_G$ 
in \emph{$G$-equivariant cohomotopy} \cite[Ex. 4.5.8]{SS25-Bund}\cite[\S 6]{SS26-Orb}\cite{SS20-Tad}
\begin{equation}
  \pi_0
  \Big(
  \mathrm{Map}\big(
    \BrillouinTorus
    ,\,
    S^2
  \big)^{\!G}
  \Big)
  \;\;
  =:
  \;\;
  \pi^2_G\big(
    \BrillouinTorus
  \big)
\end{equation}
measure the \emph{fragile symmetry-protected band topology}.

These fragile symmetry-protected topological phases have not been considered before. We highlight that there is still a stabilization map (cf. \cite[Prop. 4.5.17, Rem. 4.5.18]{SS25-Bund}) which coarsens and hence compares them to the classes in $G$-\emph{equivariant K-theory} that did find much attention in this context (cf. \cite{Stehouwer2021}).

What is more, in view of Thm. \ref{DerivingTorusAlgebra} we immediately obtain a precise formula for the algebra of topological Berry phases to be expected in symmetry-protected FQAH systems in topological phases $C_G$:
\begin{equation}
  \label{EquivariantObservables}
  G\mathrm{Obs}_{C_G}
  \;\;
  :=
  \;\;
  \mathbb{C}\bigg[
  \pi_1
  \Big(
    \mathrm{Map}\big(
      \BrillouinTorus
      ,\,
      S^3
    \big)^{\!G}
    ,\,
    C_G
  \Big)
  \bigg]
  \,.
\end{equation}
This predicts anyonic topologically ordered FQAH phases whenever the $G$-symmetry \eqref{SymmetryOnBlochHamiltonian} is such that \eqref{EquivariantObservables} is non-abelian, 

While there do not seem to be results available computing these equivariant cohomotopy algebras 
\eqref{EquivariantObservables} for non-trivial $G$-actions, in generalization of Thm. \ref{DerivingTorusAlgebra}, it is already noteworthy that hereby the problem of describing $G$-protected topological order in FQAH systems is reduced to a straightforward question in pure algebraic topology which may be handed to specialists for further investigation (cf. the end of \cite[\S 6]{SS25-Srni}).


\begin{acknowledgements}
For useful discussion we thank 
Adrien Bouhon,
Sadok Kallel, and
Robert-Jan Slager.

This research was supported by \emph{Tamkeen} under the 
\emph{NYU Abu Dhabi Research Institute grant} {\tt CG008}.
\end{acknowledgements}

\bibliography{references}

\end{document}